\documentclass[prl,twocolumn,superscriptaddress]{revtex4-2}
\usepackage{amsmath}
\usepackage{amssymb}
\usepackage{amsthm}
\usepackage{graphicx}
\usepackage{xcolor}
\usepackage[english]{babel}
\usepackage{csquotes}
\usepackage{braket}
\usepackage{times}
\usepackage[bookmarks=true,colorlinks=true,linkcolor=orange,citecolor=orange,urlcolor=orange,bookmarks]{hyperref}

\usepackage{mathtools}
\usepackage{booktabs}
\usepackage[capitalize,nameinlink]{cleveref}
\usepackage{verbatim}

\newcommand*{\pipe}{\ensuremath{\, | \, }}

\newcommand*{\fatpipe}{\ensuremath{\, \| \, }}
\newcommand*{\Tr}{\operatorname{Tr}}

\providecommand*{\poly}{\ensuremath{\operatorname{poly}}}
\newcommand*{\negl}{\ensuremath{\operatorname{negl}}}

\providecommand{\calA}{\ensuremath{\mathcal{A}}}
\providecommand{\calB}{\ensuremath{\mathcal{B}}}

\providecommand{\calD}{\ensuremath{\mathcal{D}}}
\providecommand{\calE}{\ensuremath{\mathcal{E}}}

\providecommand{\calH}{\ensuremath{\mathcal{H}}}

\providecommand{\calK}{\ensuremath{\mathcal{K}}}
\providecommand{\calL}{\ensuremath{\mathcal{L}}}

\providecommand{\calQ}{\ensuremath{\mathcal{Q}}}
\providecommand{\calR}{\ensuremath{\mathcal{R}}}

\providecommand{\calT}{\ensuremath{\mathcal{T}}}

\providecommand{\bbI}{\ensuremath{\mathbb{I}}}

\providecommand{\bbN}{\ensuremath{\mathbb{N}}}

\providecommand{\bbP}{\ensuremath{\mathbb{P}}}

\providecommand{\bbR}{\ensuremath{\mathbb{R}}}

\AtBeginDocument{
\heavyrulewidth=.08em
\lightrulewidth=.05em
\cmidrulewidth=.03em
\belowrulesep=.65ex
\belowbottomsep=0pt
\aboverulesep=.4ex
\abovetopsep=0pt
\cmidrulesep=\doublerulesep
\cmidrulekern=.5em
\defaultaddspace=.5em
}

\newtheorem{theorem}{Theorem}
\newtheorem{lemma}[theorem]{Lemma}

\newtheorem{definition}[theorem]{Definition}

\newtheoremstyle{redthm}{}{}{\color{red}}{}{\color{red}\bfseries}{}{ }{}
\theoremstyle{redthm}

\newcommand{\JJM}[1]{\textcolor{blue}{[jjm: #1]}}

\newcommand{\dccqs}{Dahlem Center for Complex Quantum Systems, Freie Universit{\"a}t Berlin, 14195 Berlin, Germany}
\newcommand{\hzb}{Helmholtz-Zentrum Berlin f{\"u}r Materialien und Energie, 14109 
Berlin, Germany}
\newcommand{\hhi}{Fraunhofer Heinrich Hertz Institute, 10587 Berlin, Germany}
\newcommand{\salerno}{Dipartimento di Ingegneria Industriale, Università degli Studi di Salerno, 84084 Fisciano (SA), Italy}

\newcommand{\nocontentsline}[3]{}

\begin{document}	
	\title{The computational two-way 
    quantum capacity}
	\date{\today}
    
        \author{J.~J.~Meyer}
        \affiliation{\dccqs}
        
        \author{J.~Rizzo}
        \affiliation{\dccqs}
        
        \author{A.~Raza}
        \affiliation{\dccqs}
        
        \author{L.~Leone}
        \affiliation{\salerno}
        \affiliation{\dccqs}
        
        \author{S.~Jerbi}
        \affiliation{\dccqs}
        \affiliation{\hzb}
                
        \author{J.~Eisert}
        \affiliation{\dccqs}
        \affiliation{\hzb}
        \affiliation{\hhi}

    \begin{abstract}
        Quantum channel capacities are fundamental to quantum information theory. Their definition, however, does not limit the computational resources of sender and receiver.
        In this work, we initiate the study of computational quantum capacities. These quantify how much information can be reliably transmitted when imposing the natural requirement that en- and decoding have to be computationally efficient.
        We focus on the computational two-way quantum capacity and showcase that it is closely related to the computational distillable entanglement of the Choi state of the channel.
        This connection allows us to show a stark computational capacity separation. Under standard cryptographic assumptions, there exists a quantum channel of polynomial complexity whose computational two-way quantum capacity vanishes while its unbounded counterpart is nearly maximal.
        More so, we show that there exists a sharp transition in computational quantum capacity from nearly maximal to zero when the channel complexity leaves the polynomial realm.
        Our results demonstrate that the natural requirement of computational efficiency can radically alter the limits of quantum communication.
    \end{abstract}

    \maketitle

Quantum information theory is concerned with identifying the potential and limitations of accomplishing tasks that involve quantum systems as carriers of information. Largely formulated in its basics in the 1990s and later refined or made rigorous, it is one of the conceptual and mathematical pillars on which quantum technologies rest~\cite{PhysRevA.54.3824,PhysRevA.55.1613,Devetak,MarkWilde,watrous2018theory}. Central to this framework are quantum channel capacities, which assess the rate at which qubits can be reliably transmitted through a noisy quantum channel and even constrains the overhead of quantum error correction~\cite{fawzi2022lower}. It was found that the quantum capacity is given by the regularized coherent information, capturing the inherently asymptotic character of the quantity
\cite{PhysRevA.55.1613,Devetak,MarkWilde,watrous2018theory}. This reflects the fact that, when sending quantum information from a sender to a receiver over multiple copies of the channel, increasingly sophisticated encoding and 
decoding procedures can be employed.

In this work, we challenge this picture, leading to unexpected and potentially surprising conclusions. We do so by adopting the premise that all steps in such a communication task must be carried out using \emph{polynomial} resources. This is a foundational insight in computer science and quantum computation: tasks are considered feasible or ``efficient'' only if the total number of steps \emph{grows polynomially} with the input size. For this reason, it is natural to require that only polynomially many copies are used and that only polynomially many quantum gates are implemented in encoding and decoding. 

\begin{figure}[tbh]
  \centering
  \includegraphics{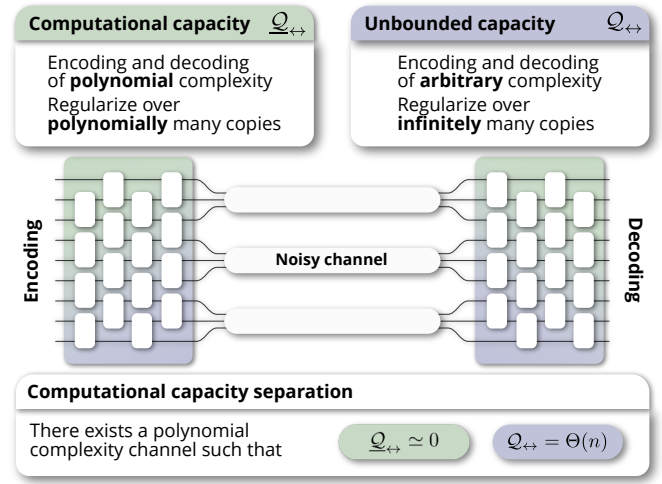}
  \caption{The computational two-way quantum capacity captures the limits of quantum information transmission with the help of bidirectional classical communication when the involved operations are of polynomial complexity. The difference to the traditional, unbounded setting can be quite dramatic, as we show there exists a channel with nearly maximal unbounded capacity but vanishing computational capacity.}
  \label{fig:complexity}
\end{figure}

Motivated by this reasoning, we initiate the study of the \emph{computational} two-way quantum capacity, that is, the two-way quantum capacity under exactly these polynomial constraints. 
With this, we establish a new direction within \emph{computational quantum information theory}~\cite{Pseudoentanglement,arnon2023computationalentanglementtheory,leone2025entanglementtheorylimitedcomputational,meyer2025computational,Yanguez2025MinMax}, ensuring that the computed capacities can be operationally realized.
We find that it is closely related to the computational distillable entanglement of the Choi state, mirroring the situation in the traditional, unconstrained setting -- and 
capturing the resource character of entanglement in a bipartite scenario. This insight 
provides a quantitative tool for assessing the computational two-way quantum capacity. Most notably, it leads to a striking separation result: Under standard cryptographic assumptions, there exist physically meaningful quantum channels of polynomial complexity whose computational two-way quantum capacity vanishes, while their traditional, unconstrained two-way quantum capacity grows with the system size. This phenomenon can be interpreted as a form of computationally bound capacity in quantum information.
To add even more, we show that there exists a sharp transition in the computational two-way quantum capacity of a generalized dephasing channel depending on its Choi-rank. In the polynomial Choi-rank regime, we provide a computationally efficient protocol that achieves nearly optimal quantum capacity, whereas the computational capacity can vanish as soon as the Choi-rank becomes super-polynomial.

The upshot of our work is that, to achieve a complete picture of quantum communication theory, 
its foundational premises must be reassessed in light of the potentially super-polynomial 
-- and hence unreasonable -- resources they implicitly might require. Now that more than thirty years have passed since the advent of quantum communication theory and quantum computation, the two fields should converge more fully, and it should become clear that notions of efficient manipulation of quantum information are central to both.
\smallskip


{\bf Computational information theory.}
Even though we enter a new conceptual realm of quantum communication, we largely follow the notation of Ref.~\cite{meyer2025computational}, where computational constraints in hypothesis testing have been studied. 
In computational (quantum) information theory, we quantify the complexity of implementing our operations and protocols relative to a scaling variable $n \in \bbN$. As such, we always deal with \emph{sequences} of objects like states, channels or Hilbert spaces. To keep the notation light, we use the convention that an index $n$ denotes a sequence of objects for $n \in \bbN$, $\rho_n$ for example denoting the sequence of states $(\rho_n)_{n \in \bbN}$.

To measure the complexity of a quantum channel $\Phi$, we fix the following computational model: operations on an input quantum state are implemented with access to an unbounded number of auxiliary systems in the $\ket{0}$ state vector. All operations are composed of elementary 2-local gates in the system of choice followed by measurements that are either projecting onto $\ket{0}$ or ignoring the system. The number of such gates is the complexity $C(\Phi)$ of the channel. Note that in this model, tracing out is a free operation. 

The above notion allows us to define gate-constrained versions of quantum information processing tasks. The most important one we use is \emph{entanglement distillation}, where our task is to use an LOCC (Local Operations and Classical Communication) map of bounded complexity to distill maximally entangled states from a shared resource state. Let $\operatorname{LOCC}(A|B;G)$ denote the set of LOCC maps relative to a bipartite system $A,B$ with complexity at most $G$ and $\phi$ be an ebit (a qubit maximally entangled state). 
The maximal amount of entanglement distillable with bounded error is then~\cite{khatri2024principles}
\begin{align}\label{eqn:def-single-shot-efficiently-distillable-entanglement}
    &E_D^{\epsilon}(\rho^{A,B}; G) = \log 2 \times \max \mathstrut\big\{ m \in \bbN \pipe \\
    &\quad \exists \Phi \in \operatorname{LOCC}(A|B; G) \colon F(\Phi[\rho^{A,B}], \phi^{\otimes m}) \geq 1-\epsilon \big\}.\nonumber
\end{align}
Here, $F(\rho, \sigma) = \lVert \sqrt{\rho}\sqrt{\sigma}\rVert_1^2$ is the quantum fidelity~\footnote{The factor $\log 2$ arises from our convention that $\log$ denotes the natural logarithm, requiring us to convert from bits to nats.}.

To get a quantity that captures the 
essence of how much entanglement can be distilled with polynomial resources, we perform a \emph{polynomial regularization} of the above quantity. The mathematical background for this is presented in Ref.~\cite{meyer2025computational}, we shall only mention the key points. In this setting, we care about the scaling of our quantities up to \emph{negligible} terms, i.e.\ terms that vanish faster than any polynomial, written as $\negl(n)$. We use the notation $f \simeq g$ to indicate that $f-g$ is a negligible function and the associated ordering relations $\gtrsim$ and $\lesssim$. Let us now give the definition of the \emph{computational distillable entanglement} of a state \smash{$\rho_n^{A_n, B_n}$}~\cite{meyer2025computational} (compare also similar notions in Refs.~\cite{Pseudoentanglement,arnon2023computationalentanglementtheory,leone2025entanglementtheorylimitedcomputational,Yanguez2025MinMax}) as
\begin{align}\label{eqn:def-computational-distllable-entanglement}
    \begin{split}
    &\underline{E}_D(\rho_n^{A_n, B_n}) :\simeq \\
    &\qquad\lim_{\epsilon \to 0}\lim_{\ell \to \infty} \liminf_{k\to\infty} \frac{1}{n^k} E_D^{\epsilon}((\rho_n^{A_n, B_n})^{\otimes n^k}; n^{k\ell}).
    \end{split}
\end{align}
The notation $:\simeq$ indicates that the right hand side is to be understood as an object in the space of functions that is obtained by modding out the equivalence relation $f \simeq g$ from the usual space of functions from $\bbN \to \bbR$. Ref.~\cite{meyer2025computational} establishes that this constitutes a metric space which in turn allows us to make sense of the involved limits, \emph{seen as limits of functions}, not as pointwise limits. The limit $k \to \infty$ takes the number of copies to an arbitrary large polynomial, which in turn allows us to neglect certain correction terms. The limit $\ell \to \infty$ takes the number of gates to an arbitrarily large polynomial, ensuring that we can implement all polynomial complexity operations. The limit $\epsilon \to 0$ is a standard part of the definition of the distillable entanglement.

Finally, we consider two quantum states $\rho_n$ and $\sigma_n$ to be \emph{computationally indistinguishable}, $\rho_n \approx_c \sigma_n$, if for all POVM effects of polynomial complexity $C(\Lambda_n) \leq O(\poly(n))$ and all $m_n \leq O(\poly(n))$, we have that 
\begin{align}
    |\Tr[\Lambda_n \rho_n^{\otimes m_n}]-\Tr[\Lambda_n \sigma_n^{\otimes m_n}]|\leq \negl(n).
\end{align}

{\bf Computational two-way quantum capacity.}
We now turn to defining the computational two-way quantum capacity. To prepare this, let us first review the standard definition. The worst-case channel fidelity of two quantum channels $\Phi$ and $\Psi$ is defined as
\begin{align}
    F(\Phi, \Psi) \coloneqq \inf_{\psi} F((\bbI \otimes \Phi)[\psi], (\bbI \otimes \Psi)[\psi]),
\end{align}
where $F$ is the fidelity of quantum states and the auxiliary system has the same dimension as the input system. To capture LOCC operations, we will make use of the concept of \emph{quantum combs}. For us, it is sufficient that the concept captures a map whose input and outputs are quantum channels and that can model adaptive computations. We use the symbol \smash{$\Phi^{[m]}$} to denote the comb that represents $m$ independent invocations of a basic channel $\Phi$.
Let now denote with $\bbI_d$ a $d$-dimensional identity channel. 
The single-shot two-way quantum capacity of a quantum channel is then defined as~\cite{khatri2024principles}
\begin{align}
\begin{split}
    \calQ_{\leftrightarrow}^\epsilon(\Phi) \coloneqq &\log \max \mathstrut \big\{ d \in \bbN \pipe \\
    \nonumber
    &\qquad \exists \calL : F(\calL [\Phi], \bbI_d)\geq 1- \epsilon \big\},
\end{split}
\end{align}
where $\calL$ is a LOCC comb (with respect to the input and output side) that represents en- and decoding. The two-way quantum capacity is then defined by regularizing as
\begin{align}
    &\calQ_{\leftrightarrow}(\Phi) \coloneqq \lim_{\epsilon \to 0} \lim_{m\to\infty} \frac1m \calQ_{\leftrightarrow}^{\epsilon}(\Phi^{[m]}).
\end{align}
Following the developments of the preceding section, we define a computational version of the two-way quantum capacity by imposing a gate restriction on the single-shot setting and then performing a polynomial regularization.
\begin{definition}[Computational two-way quantum capacity]
    We define the \emph{$G$-complexity single-shot two-way quantum capacity} of a quantum channel $\Phi$ as
    \begin{align}
    \begin{split}
    &\calQ_{\leftrightarrow}^\epsilon(\Phi; G) \coloneqq \log \max  \{ d \in \bbN \pipe \\
    &\quad \exists \calL : C(\calL)\leq G, F(\calL [\Phi], \bbI_d)\geq 1- \epsilon \},
    \end{split}
    \end{align}
    where $\calL$ is a LOCC comb.
    The \emph{computational two-way quantum capacity} of a channel $\Phi_n$ is then given by
    \begin{align}
    &\underline{\calQ}_{\leftrightarrow}(\Phi) :\simeq \lim_{\epsilon \to 0} \lim_{\ell \to \infty} \liminf_{k \to \infty} \frac{1}{n^k}\calQ_{\leftrightarrow}^{\epsilon}(\Phi_n^{[n^k]}; n^{k\ell}).
    \end{align}
\end{definition}

{\bf The capacity of efficiently stretchable channels.}
The objective of this section is to lay the foundations for our separation result by establishing a capacity formula for a quite general class of channels.
We closely follow Ref.~\cite{pirandola2017fundamental} and pay close attention to the computational efficiency of the involved operations. We start with a definition.
\begin{definition}[Efficient stretchability]
    We call a quantum channel \smash{$\Phi_n^{A_n \to B_n}$} \emph{efficiently stretchable} if there exists an LOCC map \smash{$\calL_n^{A_n, A_n'|B_n}$} of polynomial complexity \smash{$C(\calL_n^{A_n, A_n'|B_n})\leq O(\poly(n))$}, where the LOCC is relative to the bipartititon $A_n ,A_n'|B_n$, and a 
    state \smash{$\sigma_n^{A_n', B_n}$} such that
    \begin{align}
        \Phi_n^{A_n \to B_n}[X^{A_n}] = \calL_n^{A_n,  A_n'|B_n}[X^{A_n}\otimes \sigma_n^{A_n' , B_n}].
    \end{align}
    The quantum state \smash{$\sigma_n^{A_n',  B_n}$} does not need to be of polynomial complexity. We further call a quantum channel \emph{efficiently Choi-stretchable} if it is efficiently stretchable for \smash{$\sigma_n^{A_n' B_n} = J_{\Phi_n}$} the Choi state of the channel.
\end{definition}
In the following, we will usually suppress the exact identification of subsystems as it is clear from context.
It has been shown in Ref.~\cite{pirandola2017fundamental} that Choi-stretchability is implied by so-called teleportation-covariance. As a means to verify efficient Choi-stretchability, we will prove an analogous result and crucially show that this  implication preserves computational efficiency. The referenced teleportation-correction unitaries arise in the generalized teleportation protocol~\cite[Section~3.3.2]{khatri2024principles}, for our purposes it suffices to know that these operations are one-local.
\begin{lemma}[Efficient Choi-strechability]\label{lemma:efficient-stretchability-from-teleportation-covariance}
    Consider a quantum channel $\Phi_n\colon \calH_n \to \calK_n$ of polynomial complexity $C(\Phi_n) \leq O(\poly(n))$ such that $\log \dim\calH_n \leq O(\poly(n))$. If the channel is \emph{teleportation-covariant}, i.e.\ for any teleportation correction unitary $U_i$ there exists a unitary $V_i$ such that
    \begin{align}
        \Phi_n[U_i X U_i^{\dagger}] = V_i \Phi_n[X]V_i^{\dagger},
    \end{align}
    the channel is efficiently Choi-stretchable.
\end{lemma}
The proof is deferred to the end matter.
Let us now analyze the capacity of efficiently stretchable channels. We start with an achievability statement that holds for all quantum channels and gives a direct operational meaning to the computational distillable entanglement of the Choi state in the context of quantum communication. 
\begin{lemma}[Lower bound to computational two-way quantum capacity]
Consider a quantum channel $\Phi_n$ with Choi state $J_{\Phi_n}$.
Then, we have the lower bound on the computational two-way quantum capacity
\begin{align}
    \underline{\calQ}_{\leftrightarrow}(\Phi_n) \gtrsim \underline{E}_D(J_{\Phi_n}).
\end{align}
\end{lemma}
\begin{proof}
    The preparation of a maximally entangled state is an efficient operation. Sending one half through the given channel gives both parties a copy of the Choi state on which they perform the optimal efficient distillation protocol. Given the distilled ebits, they perform regular teleportation to transmit the desired information. Teleportation is 
    efficient, rendering the protocol efficient and proving the given inequality.
\end{proof}
Having established a lower bound on the computational two-way capacity that holds for all channels, we present an upper bound that holds for any efficiently stretchable channel.
\begin{lemma}[Upper bound for the computational two-way capacity]
Consider a quantum channel $\Phi_n$ that is efficiently stretchable as $\Phi_n[X] = \calL_n[X \otimes \sigma_n]$. Then, we have the upper bound on the computational two-way quantum capacity
\begin{align}
    \underline{\calQ}_{\leftrightarrow}(\Phi_n) \lesssim \underline{E}_D(\sigma_n).
\end{align}
\end{lemma}
\begin{proof}
    This follows directly from Ref.~\cite{pirandola2017fundamental}, where it was shown that the two-way LOCC comb $\calL_0$ that realizes the transmission over $m$ copies of the channel can be combined with the $m$ LOCC operations that realize the channel $\calL_n$ to get a composite LOCC map $\calL_n'$ that acts on $m$ copies of the resource state $\calL_n'[ \sigma_n^{\otimes m} ]$.
    Importantly, the complexity of $\calL_n'$ fulfills
    \begin{align}
        C(\calL_n') \leq C(\calL_0) + m C(\calL_n),
    \end{align}
    which is polynomial if all involved operations have polynomial complexity and we work with polynomially many copies.
    As distilling ebits and transmitting qubits is equivalent, this means that the two-way quantum capacity is at most the distillable entanglement that would be achieved with full control over $\calL'$.
\end{proof} 
These conditions coincide for efficiently Choi-stretchable channels, which allows us to completely characterize their computational two-way quantum capacity.
\begin{theorem}[Characterization of the two-way capacity]\label{theorem:two-way-capacity}
Consider a quantum channel $\Phi_n$ that is efficiently Choi-stretchable as $\Phi_n[X] = \calL_n[X \otimes J_{\Phi_n}]$.
Then, the computational two-way quantum capacity is given by
\begin{align}
    \underline{\calQ}_{\leftrightarrow}(\Phi_n) \simeq \underline{E}_D(J_{\Phi_n}).
\end{align}
\end{theorem}
\begin{proof}
    Immediate from the above two lemmas.
\end{proof}

{\bf Separation of computational and unbounded capacity.}
We now turn to describing the striking separation mentioned earlier. For a given probability distribution $p(x)$ over bitstrings $x \in \{0,1\}^n$, let us use the shorthand $Z^x = Z^{x_1} \otimes Z^{x_2} \otimes \dots \otimes Z^{x_n}$ to define a generalized dephasing channel as 
\begin{align}
    \calD_p[\rho] &\coloneqq \sum_{x \in \{0,1\}^n } p(x) Z^x \rho Z^x.
\end{align}
We note that the above is a Pauli channel and, as such, \emph{efficiently Choi stretchable}~\cite{pirandola2017fundamental}. The efficiency directly follows from the fact that a standard teleportation protocol with the maximally entangled state replaced by the Choi state of the Pauli channel is enough to simulate it~\cite{bowen2001teleportation}. This connection allows us to deduce the following formula for the two-way quantum capacity, see also Refs.~\cite{wilde2017converse,khatri2024principles}.
\begin{lemma}[Two-way quantum capacity of the dephasing quantum channel]\label{lemma:capacity-of-generalized-dephasing}
    Consider the generalized dephasing channel $D_p$. Then,
    \begin{align}
        \calQ_{\leftrightarrow}(\calD_p) = E_D(J_{\calD_p}) = D(p \fatpipe u).
    \end{align}
\end{lemma}
We note that, analogously to the result of Ref.~\cite{lami2023exact}, all capacities between quantum and strong-converse private are actually equal for this class of channels.

\begin{proof}[Proof outline]
    As observed in Ref.~\cite{pirandola2017fundamental}, the reasoning that led us to \cref{theorem:two-way-capacity} works equally well in the unbounded setting, and shows equality of unbounded two-way quantum capacity and distillable entanglement of the Choi state. We then lower-bound the two-way quantum capacity by the single-copy coherent information evaluated on a maximally entangled input state vector, whose entropy is equal to that of the distribution $p$ because different Paulis map the maximally entangled state to orthogonal states. We then upper-bound the distillable entanglement using the Rains bound, where we use the Choi state of the fully dephasing channel as a candidate state. The matching upper-bound is then obtained realizing that for any distribution $p$, $\calD_p[\rho] = \calT[\rho \otimes p]$ for a map $\calT$ that simply measures the classical register and applies the corresponding Pauli on the input register. This allows us to conclude via data processing.
    The full proof is relegated to the end matter.
\end{proof}
Having established the capacity of our generalized dephasing channel in the unbounded setting, we can exploit our formula for the computational two-way quantum capacity to show a stark separation. 
\begin{theorem}[Computational capacity separation]\label{theorem:computational-capacity-separation}
    Assuming the existence of quantum-secure one-to-one one-way functions, for all $\delta >0$, there exists a quantum channel of polynomial complexity $\calD_n$ such that
    \begin{align}
        \underline{Q}_{\leftrightarrow} (\calD_n) \simeq 0, \text{ but } {Q}_{\leftrightarrow}(\calD_n) = n \log 2 - \Omega(n^\delta).
    \end{align}
\end{theorem}
\begin{proof}[Proof outline]
    We use the fact that the Choi states associated to channels defined via computationally indistinguishable distributions inherit this computational indistinguishability. We then show that under the cryptographic assumptions of the theorem, there exists a distribution computationally indistinguishable from the uniform distribution with entropy $n^\delta$, which proves the unbounded two-way quantum capacity statement. The proof is then concluded by showing that the computational distillable entanglement has the computational smoothing property, which allows us to use the computational indistinguishability of the Choi states to conclude that the computational distillable entanglement is at most that of the Choi state of the fully dephasing channel, which is zero.
    The full proof is relegated to the end matter. 
\end{proof}

{\bf Sharp transition of computational capacity.}
We now complement the separation between computational and unbounded capacity with another remarkable phenomenon: a sharp transition in the computational capacity of a generalized dephasing channel from nearly maximal to zero, depending on the size of the support of the involved probability distribution. For any subset $S_n \subseteq \{0,1\}^n$ let us define $u_{S_n}$ to be the uniform distribution over bitstrings in $S_n$ and denote for brevity $\calD_{S_n} = \calD_{u_{S_n}}$. With this notation we can show the following.
\begin{theorem}[Sharp transition in computational capacity]\label{theorem:sharp-computational-capacity-transition}
    Consider the quantum channel $\calD_{{S_n}}$. Then, for any $S_n$ such that $|S_n| \leq O(\poly(n))$, we have
    \begin{align}
    \underline{\calQ}_{\leftrightarrow}(\calD_{{S_n}}) \gtrsim n \log 2 - \omega(\log n)
    \end{align}
    which is near optimal as $\calQ_{\leftrightarrow}(\calD_{{S_n}}) = n \log 2 - O(\log n)$ in this case. However, assuming the existence of quantum-secure one-to-one one-way functions, there exists a subset $S_n'$ with $|S_n'| = \omega(\poly(n))$ such that 
    \begin{align}
        \underline{\calQ}_{\leftrightarrow}(\calD_{{S_n'}}) \simeq 0.
    \end{align}
\end{theorem}
\begin{proof}[Proof outline]
We give a protocol 
that achieves the computational capacity of the theorem statement that is efficient as long as $|S_n|\leq O(\poly(n))$. The protocol starts with an application of Hadamard gates so that we work with bitflips $X^x$ instead of phase flips $Z^x$. We then use the fact that the Choi state of this generalized bitflip channel is a mixture of maximally entangled states, where every copy carries an unknown correction $X^x$. Alice and Bob apply a random Clifford such that $X^x$ is mapped to $X^{\pi(x)}$, where $\pi(x)$ is a uniformly random bitstring different from the identity. They then measure the first $m = \omega(\log n)$ bits, and we show that this information is enough to reliably find $\pi(x)$ inside of $\pi(S_n)$, in a way that is efficient as long as $|S_n| \leq O(\poly(n))$. Alice and Bob then undo $X^{\pi(x)}$ and are left with $n - \omega(\log n)$ many ebits. 
The transition for $|S_n| = \omega(\poly(n))$ is a consequence of \cref{theorem:computational-capacity-separation}, observing that the distribution used in this proof is of the required form. 
The full proof is relegated to the end matter.
\end{proof}

{\bf Outlook.} 
In this work, we have challenged established notions of quantum communication by accepting that all steps being taken must be efficient to end up at a meaningful notion of a quantum capacity. It is the hope that the present work helps advocating a paradigm of efficient quantum information theory that accepts that at the end of the day, only efficient protocols can be implemented -- bringing notions of quantum computing and communication a lot closer together.\smallskip

{\bf Acknowledgments.}
This project has been funded by the 
BMFTR (QR.N, DAQC, QSol, MuniQC-Atoms, Hybrid++), 
the BMWK (EniQmA),
the Munich Quantum Valley (K-8),
the Clusters of Excellence (ML4Q and MATH+),
Berlin Quantum, the QuantERA (HQCC),
the Quantum Flagship (Millenion and Pasquans2),
and the European Research Council 
(DebuQC).

\bibliography{main}

\clearpage
\appendix

\section{End Matter}

\begin{proof}[Proof of \cref{lemma:efficient-stretchability-from-teleportation-covariance}]
    The fact that $\Phi_n$ is Choi-stretchable is already shown in Ref.~\cite{pirandola2017fundamental}. We are left to argue that $V_i$ necessarily has polynomial complexity.
    As $\Phi_n$ has polynomial complexity in our computational model, there exists a unitary dilation $W_n$ of polynomial complexity such that
    \begin{align}
        \Phi_n[X] = \Tr_{\mathrm{anc}}[W_i(|0\rangle\!\langle 0|_{\mathrm{anc}} \otimes X)W_i^{\dagger}].
    \end{align}
    Using the dilation we get the equality
    \begin{align}
        &V_i \Phi_n[X]V_i^{\dagger} \\
        \nonumber
        &= V_i \Tr_{\mathrm{anc}} [W_i(|0\rangle\!\langle 0|_{\mathrm{anc}} \otimes X)W_i^{\dagger}] V_i^\dagger\\
        \nonumber&=\Tr_{\mathrm{anc}} [(Q_i \otimes V_i)W_i(|0\rangle\!\langle 0|_{\mathrm{anc}} \otimes  X )W_i^{\dagger}(Q_i \otimes V_i)^\dagger], 
        \nonumber
    \end{align}
    where we used the fact that $V_i$ only acts on the systems that are not traced out and $Q_i$ is any unitary operator. Similarly, from the other direction
    \begin{align}
        &\Phi_n[U_i X U_i^{\dagger}] \\
        \nonumber
        &=\Tr_{\mathrm{anc}}[W_i(|0\rangle\!\langle 0|_{\mathrm{anc}} \otimes U_i X U_i^{\dagger})W_i^{\dagger}] \\
         \nonumber
        &= \Tr_{\mathrm{anc}}[\tilde{U}_i W_i(|0\rangle\!\langle 0|_{\mathrm{anc}} \otimes X)W_i^{\dagger}\tilde{U}_i^\dagger],
         \nonumber
    \end{align}
    where $\tilde{U}_i^{\dagger} = W_i (\bbI \otimes U_i) W_i^\dagger$.
    As both expressions need to be equal, there must exist some $Q_i$ such that
    \begin{align}
        W_i (\bbI \otimes U_i) W_i^\dagger = Q_i \otimes V_i.
    \end{align}
    As $V_i$ can hence be implemented as
    \begin{align}
        V_i = \Tr_{\mathrm{anc}}[W_i (\bbI \otimes U_i) W_i^\dagger],
    \end{align}
    we have the upper bound $C(V_i) \leq C(U_i) + 2C(W_i) = C(U_i) + 2 C(\Phi_n)$. By assumption, $C(\Phi_n) \leq O(\poly(n))$. Teleportation-correction unitaries are always one-local operations which can be implemented efficiently if $\log \dim \calH_n \leq O(\poly(n))$, which holds by assumption. Hence, $C(V_i) \leq O(\poly(n))$ as desired.
\end{proof}

\begin{proof}[Proof of \cref{lemma:capacity-of-generalized-dephasing}]
    As observed in Ref.~\cite{pirandola2017fundamental}, the reasoning that led us to \cref{theorem:two-way-capacity} works equally well in the unbounded setting, which proves that
    \begin{align}
        \calQ_{\leftrightarrow}(\calD_p) = E_D(J_{\calD_p}).
    \end{align}
    To establish the exact value of the two-way capacity, we follow Ref.~\cite{lami2023exact} and lower-bound the two-way quantum capacity by the coherent information evaluated on a maximally entangled input state vector
    \begin{align}
        \gamma^{A,B} = \frac{1}{2^n} \sum_{x,x' \in \{0,1\}^n} |x,x\rangle\!\langle x',x'|.
    \end{align}
    This yields
    \begin{eqnarray}
        \calQ_{\leftrightarrow}(\calD_p) &\geq &\calQ(\calD_p) 
        \geq S(\calD_p[\gamma^{A}]) - S((\calD_p \otimes \bbI)[\gamma^{A B}])\nonumber \\
        &= &n \log 2 - S(J_{\calD_p}).
    \end{eqnarray}
    The entropy of the Choi state simply evaluates to
    \begin{align}
        S(J_{\calD_p}) &= (\calD_p \otimes \bbI)[\gamma^{A_n, B_n}] \\
         \nonumber
        &= S\Bigg(\sum_{x \in \{0,1\}^n} p(x) Z^x \gamma^{A_n, B_n} Z^x\Bigg) 
        = S(p),
    \end{align}
    because all states $Z^x \gamma Z^x$ are orthogonal to each other. Recognizing that $n\log 2 - S(p) = D(p\fatpipe u)$ settles the lower bound. To claim the upper bound, we use Rains bound on the distillable entanglement, which posits that $E_D(\rho^{A,B}) \leq \min_{\sigma^{A,B} \text{ PPT}}D(\rho^{A,B} \fatpipe \sigma^{A,B})$, where the relative entropy is minimized over all \emph{positive partial transpose} (PPT) states. A particular instance of a PPT state is the maximally classically correlated state, the Choi state of the fully dephasing channel, 
    \begin{align}
        \Bar{\gamma} = \frac{1}{2^n} \sum_{x \in \{0,1\}^n} |x,x\rangle\!\langle x,x| = J_{\calD_u}.
    \end{align}
    We now use the fact that we can implement the generalized dephasing channel as $\calD_p[\rho] = \calT[\rho \otimes p]$, where $\calT$ simply measures the classical register in the computational basis and upon obtaining the outcome $x$ implements $Z^x$ on the input register. Then, we have
    \begin{align}
        E_D(J_{\calD_p}) &\leq D(J_{\calD_p} \fatpipe J_{\calD_u}) \\
         \nonumber
        &= D((\bbI \otimes \calD_p)[\gamma] \fatpipe (\bbI \otimes \calD_u)[\gamma]) \\
         \nonumber
        &= D(\bbI \otimes \calT)[\gamma \otimes p] \fatpipe (\bbI \otimes \calT)[\gamma \otimes u]) \\
         \nonumber
        &\leq D(\gamma \otimes \fatpipe \gamma \otimes u)  
        = D(p\fatpipe u).
         \nonumber
    \end{align}
    Here we have just combined the above simulation argument with data-processing of the relative entropy to conclude the proof of the lemma. The comment below follows from us having lower-bounded the quantum capacity directly and the fact that the upper bound in terms of the distillable entanglement holds also for the strong converse private capacity 
    \cite{wilde2017converse}.
\end{proof}
\begin{proof}[Proof of \cref{theorem:computational-capacity-separation}]
    Let us first consider the generalized dephasing channel associated to two classical distributions $p_n$ and $q_n$ that are computationally indistinguishable. In this case, the Choi states of the associated channels inherit the computational indistinguishability
    \begin{align}
        p_n \approx_c q_n \ \Rightarrow \ J_{\calD_{p_n}} \approx_c J_{\calD_{q_n}}.
    \end{align}
    This follows because we can implement the channel $\calD_{p_n}$ with access to the distribution $p_n$ by first measuring it in the computational basis, obtaining outcome $x$, and subsequently implementing $Z^x$ on the input quantum state. Upon input of the maximally entangled state, which is efficiently preparable, we obtain the Choi state. Hence, any protocol for distinguishing the Choi states could be used to distinguish the distributions, yielding a contradiction.
    Next, we need the fact that under the assumptions on the existence of one-way functions imposed in the Theorem statement, there exist a classical distribution $p_n$ over length $n$ bitstrings that is computationally indistinguishable from the uniform distribution $u_n$ but at the same time fulfills 
    \begin{align}
        D(p_n \fatpipe u_n) = n \log 2 - \Omega(n^\delta).
    \end{align}
    We use here a very standard construction based on the Goldreich-Levin Theorem (see, e.g., Ref~\cite{srinivasan2023fundamentals}). This theorem notably allows us to construct a pseudorandom generator $G: \{0,1\}^m \rightarrow \{0,1\}^{m+1}$ such that the distribution $\{G(x)\}_{x \leftarrow u_m}$ is computationally indistinguishable from the uniform distribution $u_{m+1}$, i.e., it stretches the $m$-bit uniform distribution by $1$ bit. We can then apply this construction recursively, to extend this stretch to arbitrary $\poly(m)$ bits, under the same hardness assumption of inverting the one-way function (see Corollary 1 in Ref~\cite{srinivasan2023fundamentals}). Since we want a pseudorandom distribution $p_n$ on $n$ bits, we can therefore choose $m = n^\delta$ for any $\delta>0$ and stretch $u_m$ polynomially to $n$ bits. This construction only uses $n^\delta$ bits of randomness, which gives the desired upper bound on the entropy of $p_n$.
    The final ingredient we need is to show that the efficiently distillable entanglement $\underline{E}_D$ has the \emph{computational smoothing} property~\cite{meyer2025computational}. Indeed, consider any state $\Tilde{\rho}_n^{A_n, B_n}$ computationally indistinguishable from $\rho_n^{A_n, B_n}$. Then, the objective in the definition of the single-shot distillable entanglement of \cref{eqn:def-single-shot-efficiently-distillable-entanglement} reads
    \begin{align}
        F(\Phi[\rho^{A B}, \phi^{\otimes m}]
        &= \Tr[ \Phi[\rho^{A B}]\phi^{\otimes m} ] 
        =\Tr[ \rho^{A B} \Phi^{\dagger}[\phi^{\otimes m} ]]. 
    \end{align}
    When performing the regularization to obtain the computationally distillable entanglement of \cref{eqn:def-computational-distllable-entanglement}, we observe that the complexity of the POVM effect $\Phi^{\dagger}[\phi^{\otimes m}]$ is always polynomial as the maximally entangled state $\phi$ is efficiently preparable and $\Phi$ is an efficient LOCC map. Therefore, computational indistinguishability posits that any protocol that guarantees a fidelity of at least $1-\epsilon$ upon input of $\rho_n^{A_n, B_n}$ also has to guarantee a fidelity of $1-\epsilon - \negl(n)$ for $\tilde\rho_n^{A_n, B_n}$. For $n$ sufficiently large, we always have that the negligible term is $\leq \epsilon$ and upon taking the final limit $\epsilon \to 0$ we recover
    \begin{align}
        \underline{E}_D(\rho_n^{A_n, B_n}) = \underline{E}_D(\tilde\rho_n^{A_n, B_n}).
    \end{align}
    As the choice of $\tilde\rho_n^{A_n, B_n}$ among all computationally indistinguishable states was arbitrary, we obtain the desired computational smoothing property
    \begin{align}
        \underline{E}_D(\rho_n^{A_n, B_n}) = \min_{\tilde\rho_n^{A_n, B_n} \approx_c \rho_n^{A_n, B_n}} \underline{E}_D(\tilde\rho_n^{A_n, B_n}).
    \end{align}
    Therefore, we can reason that
    \begin{align}
        \underline{E}_D(J_{\calD_{p_n}}) \simeq \underline{E}_D(J_{\calD_{u_n}}) \simeq 0,
    \end{align}
    the latter equality following because the Choi state of the full dephasing channel associated to the uniform distribution $u_n$ is not entangled.
    Given theses results, the theorem statement then follows from the capacity formula in the unbounded setting of \cref{lemma:capacity-of-generalized-dephasing} and the link between computational capacity and distillable entanglement of \cref{theorem:two-way-capacity}.
\end{proof}

\begin{proof}[Proof of \cref{theorem:sharp-computational-capacity-transition}]
    We shall prove an achievability statement for the distillable entanglement of the Choi state $J_{D_{{S_n}}}$ in the case $|S_n| = O(\poly(n))$. First of all, Alice and Bob can apply a layer of Hadamard gates to all their qubits, transforming the operators $Z^x$ in the definition of $\calD_{S_n}$ to bitflips $X^x$. We can then think of the Choi state as a uniform mixture of states of the form $(X^x \otimes \bbI)\ket{\gamma}$ for $x \in S_n$, where $\ket{\gamma}$ is the maximally entangled state vector. Our strategy is now to find out $x$ without disturbing too many qubits.
    To do so, similarly to Ref.~\cite{gu2025constantoverheadentanglementdistillation}, Alice and Bob agree on a uniformly random Clifford operation $C_\pi$, where Alice implements $C_\pi$ on her half of the Choi state and Bob implements $C_\pi^T$ on his half, realizing a permutation of the bitstrings $x \mapsto \pi(x)$ such that $\pi(x)$ is uniform over all bitstrings distinct from the all-zero one. As Cliffords can be implemented in polynomial time, the communication cost is also polynomial. The uniformity of the image $\pi(x)$ is a consequence of group structure of Cliffords, which means that the probability that $C_\pi X^x C_\pi^{\dagger}=X^{x'}$ (for $x,x'\neq 0$) is left- and right invariant under multiplication with the Clifford that would map $x$ to $x'$, hence proving the uniformity of the induced permutation of bitstrings.
    Now, Alice and Bob both measure the first $m$ qubits of their half of the thus transformed Choi state, obtaining measurement outcomes $y_A$ and $y_B$ such that $y = y_A + y_B$ equals the first $m$ bits of $y = \pi(x)|_m$. 
    We claim that the knowledge of $y$ is sufficient to identify $\pi(x)$ from the set $\pi(S_n)$ with high probability. To see this, consider that for $x \neq x'$ we have that
    \begin{align}        \operatornamewithlimits{\bbP}_{C_\pi}\big[ \pi(x)|_m = \pi(x')|_m \big] = \frac{2^{n-m}-1}{2^n-1} \leq 2^{-m}.
    \end{align}
    This bound arises by simply counting the number of bitstrings (excluding the all-zero one) that have a fixed pattern on the first $m$ bits.
    The union bound then posits that the probability of a collision which would make it impossible to reliably identify the bitstring inside $\pi(S_n)$ is at most $|S_n|^2 2^{-m}$. We note that the failure case can be absorbed into the $\epsilon$ parameter in the definition of the distillable entanglement, which means we only need to guarantee a failure probability of $O(\epsilon/n^k)$, where $n^k$ is the number of available copies. This is surely achieved by choosing
    \begin{align}
        m = \log \frac{|S_n|^2}{\negl(n)} = \omega(\log n)
    \end{align}
    many bits to measure. 
    From the observation $y$, Alice then performs a search inside $\pi(S_n)$ which can be efficiently done exactly if $|S_n|=O(\poly(n))$.
    Having reliably identified $\pi(x)$ from $\pi(S_n)$, Alice applies the operator $X^{\pi(x)}$ and Alice and Bob now share $n-m = n - \omega(\log n)$ ebits, proving the achievability statement. The corresponding achievability in the unbounded setting is a direct consequence of \cref{lemma:capacity-of-generalized-dephasing}. The fact that there exist a channel with $|S_n| = \omega(\poly(n))$ of vanishing computational capacity is part of \cref{theorem:computational-capacity-separation}, observing that the pseudorandom distribution used in the construction is of the form $u_{S}$ for some set $S$.
\end{proof}

\end{document}